\documentclass{article}
\usepackage{graphicx}
\usepackage{amsmath, amsthm, latexsym}
\usepackage{amssymb}
\newtheorem{theorem}{Theorem}

\newtheorem{definition}{Definition}
\newtheorem{corollary}{Corollary}
\newtheorem{proposition}{Proposition}
\newtheorem{remark}{Remark}

\begin{document}

\vspace*{3cm} \thispagestyle{empty}
\vspace{5mm}

\noindent \textbf{\Large Grand canonical ensembles in general relativity}\\

\noindent  \textbf{\normalsize David Klein}\footnote{Department of Mathematics and Interdisciplinary Research Institute for the Sciences, California State University, Northridge, Northridge, CA 91330-8313. Email: david.klein@csun.edu.}
\textbf{\normalsize and Wei-Shih Yang}\footnote{Department of Mathematics, Temple University, Philadelphia, PA 19122. Email: yang@temple.edu.}\\

\vspace{4mm} \parbox{11cm}{\noindent{\small We develop a formalism for general relativistic, grand canonical ensembles in space-times with timelike Killing fields. Using that, we derive ideal gas laws, and show how they depend on the geometry of the particular space-times. A systematic method for calculating Newtonian limits is given for a class of these space-times, which is illustrated for Kerr space-time. In addition, we prove uniqueness of the infinite volume Gibbs measure, and absence of phase transitions for a class of interaction potentials in anti-de Sitter space.}\vspace{5mm}\\


}\\
\vspace{6cm}
\pagebreak

\setlength{\textwidth}{27pc}
\setlength{\textheight}{43pc}
\noindent \textbf{{\normalsize 1. Introduction}}\\

\noindent Newtonian infinite volume grand canonical  Gibbs measures on Riemannian manifolds have been studied for more than a decade \cite{rockner, albeverio, kuna, kondratiev}.  Recently, in the context of general relativity, the finite volume canonical ensemble was developed and compared to Newtonian analogs for an ideal gas of test particles, provided that the container of gas is within the range of validity of Fermi coordinates for an observer whose worldline follows a timelike Killing vector, and the container is stationary relative to the observer \cite{KC2}. \\

\noindent Here we consider grand canonical ensembles, with the assumption that the test particles constituting the gas or fluid do not affect the background metric tensor.  One might expect that the theory of specifications and Gibbs states for Riemannian manifolds developed in \cite{rockner, albeverio, kuna}, and the references therein, could be applied in a straightforward manner to Lorentzian space-time manifolds, thus providing a framework for general relativistic statistical mechanics.  However, there are several impediments.  First, space-time and energy-momentum coordinates are not easily decoupled as in the Newtonian case.  One must work from the start with the cotangent bundle of space-time.  In some spacetimes there is a nonzero energy of the vacuum that must be accounted for in a thermodynamic formalism.  There are also ambiguities associated with the notion of a thermodynamic limit, since space-time can be foliated into space slices along an observer's timelike path in many different ways. In addition, the notion of potential energy, which is a feature of central importance in statistical mechanics, normally has no place in general relativity. These issues are considered in the sequel.\\

\noindent For the purposes of statistical mechanics, our requirement that the spacetime possess a timelike Killing field is natural. The Killing vector field, tangent to the timelike path of an observer, determines a time coordinate for the system, and thus an energy component of the four-momentum.  As a consequence, comparisons to corresponding non relativistic statistical mechanical formulas are possible in some space-times for certain observers.  Moreover, this choice of four-velocity forces the particle system to evolve in such a way that the geometry of spacetime is unchanging along its worldsurface (since the Lie derivative of the metric vanishes in that direction), making it plausible on physical grounds for the  system to reach equilibrium.\\

\noindent The significance of the frame of reference in relativistic statistical mechanics was discussed in \cite{KC2}. Even in Minkowski space-time, there is no Lorentz invariant thermal state for an ideal gas system. A thermal state is in equilibrium only in  a preferred Lorentz frame in which the container of gas is at rest, and therefore breaks Lorentz invariance \cite{rovelli}. The natural generalization for statistical mechanics in a space-time with a timelike Killing field is an orthonormal tetrad in which the Killing vector determines the time axis. \\

\noindent In Sect. 2 of this paper, we develop a non rotating, orthonormal ``Fermi-Walker-Killing" coordinate system in which the Killing vector serves as the time axis.  The measure on phase space based on these coordinates satisfies a Liouville theorem and is invariant under space coordinate transformations.  Sect. 3 applies the grand canonical formalism already available for Riemannian manifolds to the Fermi surface defined in the previous section, and establishes notation.  Sect. 4 provides a general relativistic ideal gas law. The term ``ideal gas'' is somewhat misleading in the context of general relativity. This is because a volume of gas subject to no forces is still affected by the curvature of space-time, and this corresponds to a Newtonian gas subject to gravitational, tidal, and in some instances ``centrifugal forces,'' but otherwise ``ideal.''\\  

\noindent In Sect. 5, we define ``Newtonian limit'' as the limit of statistical mechanical quantities as the speed of light $c\rightarrow\infty$ for a class of spacetimes whose timelike Killing fields have a certain dependence on $c$.  We prove convergence of relativistic Gibbs states evaluated on cylinder sets to their Newtonian counterparts for the ideal gas case, and define associated Newtonian pressures.  This Newtonian limit is illustrated in detail in Sect. 6 for a gas of particles following a circular orbit in Kerr space-time.  In Sect. 7, we prove uniqueness of the infinite volume Gibbs state in anti-de Sitter space for a class of equilibrium interaction potentials and calculate the infinite volume relativistic pressure for an ideal gas of test particles.  Concluding remarks are given in Sect. 8.

\section*{\normalsize{2. Particle systems on space-time manifolds}}

\noindent Let $(\mathcal{M}, g)$ be a smooth four-dimensional Lorentzian manifold.  The Levi-Civita connection is denoted by $\nabla$, and throughout we use the sign conventions of Misner, Thorne and Wheeler \cite{MTW}. Let $\sigma(\tau)$ be a a timelike path parameterized by proper time $\tau$ with unit tangent vector $e_{0}(\tau)$.  The path, $\sigma(\tau)$,  represents the worldline of an observer, who (in principle) may make measurements. \\  

\noindent Measurements in a gravitational field are most easily interpreted through the use of a system of locally inertial coordinates. For the observer  $\sigma(\tau)$, Fermi-Walker coordinates provide such a system. A Fermi-Walker coordinate frame is nonrotating in the sense of Newtonian mechanics and is realized physically as a system of gyroscopes \cite{MTW, walker, synge, CK2}.  Expressed in these coordinates, the metric along the path is Minkowskian, with first order corrections away from the path that depend only on the four-acceleration of the observer.  If the path is geodesic, the coordinates are commonly referred to as Fermi or Fermi normal coordinates, and the metric is Minkowskian to first order near the path with second order corrections due only to curvature of the space-time \cite{MTW}. A partial listing of the numerous applications of Fermi-Walker coordinates is given in \cite{KC3} .\\

\noindent A Fermi-Walker coordinate system for $\sigma$ may be constructed as follows. A vector field $\vec{v}$ in $\mathcal{M}$ is said to be Fermi-Walker transported along $\sigma$ if $\vec{v}$ satisfies the Fermi-Walker equations given in coordinate form by,

\begin{equation}
F_{e_{0}}(v^{\alpha})\equiv \nabla_{e_{0}}\;v^{\alpha}+\Omega^{{\alpha}}_{\;\,\beta}v^{\beta}=0\,\label{F1}.
\end{equation}
 
\noindent where $\Omega^{{\alpha}}_{\;\,\beta}=a^{\alpha}u_{\beta}- 
u^{\alpha} a_{\beta}$, and $a^{\alpha}$ is the four-acceleration along $\sigma$. A Fermi-Walker coordinate system along $\sigma$ is determined by an orthonormal tetrad of vectors, $e_{0}(\tau)$, $e_{1}(\tau), e_{2}(\tau), e_{3}(\tau)$ Fermi-Walker transported along $\sigma$. Fermi-Walker coordinates $x^{0}$, $x^{1}$, $x^{2}$, $x^{3}$ relative to this tetrad   are defined by,

\begin{equation}\label{F2}
\begin{split}
x^{0}\left (\exp_{\sigma(\tau)} (\lambda^{j}e_{j}(\tau)\right)&= \tau \\
x^{k}\left (\exp_{\sigma(\tau)} (\lambda^{j}e_{j}(\tau)\right)&= \lambda^{k}, 
\end{split} 
\end{equation}

\noindent where here and below, Greek indices run over $0,1,2,3$ and Latin over $1,2,3$. The exponential map, $\exp_{p}(\vec{v})$, denotes the evaluation at affine parameter $1$ of the geodesic starting at the point $p$ in the space-time, with initial derivative $\vec{v}$, and it is assumed that the $\lambda^{j}$ are sufficiently small so that the exponential maps in Eq.\eqref{F2} are defined.  From the theory of differential equations, a solution to the geodesic equations depends smoothly on its initial data so it follows from Eq.\eqref{F2} that Fermi-Walker coordinates are smooth. Moreover, it follows from \cite{oniell} that there exists a neighborhood of $\sigma$ on which the map $\chi =(x^{0}, x^{1}, x^{2}, x^{3})$ is a diffeomorphism onto an open set in $\mathbb{R}^{4}$ and hence a coordinate chart. General formulas in the form of Taylor expansions for coordinate transformations to and from  Fermi-Walker coordinates, are given in \cite{KC1}. \\

\noindent  Assume now that $\vec{K}$ is a timelike Killing vector field in a neighborhood of $\sigma$ and the tangent vector to $\sigma$ is $\vec{K}$, i.e., $e_{0} = \vec{K}$.  The vector field $\vec{K}$ is the infinitesimal generator of a local one-parameter group $\phi_{s}$ of diffeomorphisms. The function $\phi_{s}$ is the flow with tangent vector $\vec{K}$ at each point, and $\sigma(\tau) = \phi_{\tau}(0,0,0,0)$ (where $(0,0,0,0)$ is the origin in Fermi-Walker coordinates).\\

\noindent Following \cite{KC2}, define a diffeomorphism $\bar{\chi}^{-1}$ from a sufficiently small neighborhood of the origin in $\mathbb{R}^{4}$ to a neighborhood $U$ of $\sigma(0)$ in $\mathcal{M}$ by

\begin{equation} \label{k1}
\bar{\chi}^{-1}(\bar{x}^{0}, \bar{x}^{1}, \bar{x}^{2}, \bar{x}^{3}) \equiv \phi_{\bar{x}^{0}}(\chi^{-1}(0, \bar{x}^{1}, \bar{x}^{2}, \bar{x}^{3}))
\end{equation}

\noindent Then $\bar{\chi} =(\bar{x}^{0}, \bar{x}^{1}, \bar{x}^{2}, \bar{x}^{3})$ is a coordinate system on $U$ which, as in \cite{KC2}, we refer to as Fermi-Walker-Killing coordinates. It follows that,

\begin{eqnarray} \label{k2}
\begin{split}
(0, \bar{x}^{1}, \bar{x}^{2}, \bar{x}^{3}) &= (0, x^{1}, x^{2}, x^{3})\\
\frac{\partial}{\partial \bar{x}^{i}}\Bigr|_{\bar{x}^{0} =0}&=\frac{\partial}{\partial x^{i}}\Bigr|_{x^{0} =0}\\
\frac{\partial}{\partial \bar{x}^{0}} &= \vec{K}.
\end{split}
\end{eqnarray}

\noindent Thus, in Fermi-Walker-Killing coordinates, the time coordinate $\bar{x}^{0}$ is the parameter of the flow generated by the Killing field $\vec{K}$ with initial positions of the form $(0, x^{1}, x^{2}, x^{3})$ in Fermi-Walker coordinates.\\

\noindent In the case of Fermi-Walker coordinates we designate momentum form coordinates as $p_{\alpha}$ and in the case of Fermi-Walker-Killing coordinates, the momentum form coordinates will be designated as $\bar{p}_{\alpha}$.  Thus, coordinates of the cotangent bundle of $\mathcal{M}$ will be represented as $\{ x^{\alpha},p_{\beta}\}$ or $\{ \bar{x}^{\alpha},\bar{p}_{\beta}\}$.  The metric components in Fermi-Walker coordinates are designated as $g_{\alpha\beta}$ and in Fermi-Walker-Killing coordinates as $\bar{g}_{\alpha\beta}$. Note that $\bar{g}_{\alpha\beta}$ does not depend on $\bar{x}^{0}$ because of Eqs.\eqref{k2}.\\

\noindent The state of a single  particle
consists of its four space-time coordinates together  with its four-momentum coordinates
$\{ \bar{x}^{\alpha},\bar{p}_{\beta}\}$, but the one-particle Hamiltonian satisfies,

\begin{equation}
H=\frac{1}{2}\bar{g}^{\alpha\beta}\bar{p}_{\alpha}\bar{p}_{\beta}=-\frac{m^{2}c^{2}}{2},
\label{a24}
\end{equation}

\noindent where, for timelike geodesics, $m$ is the rest mass of the particle, and $c$ is the speed of light. Eq.\eqref{a24} just states the standard fact from relativity that the square of the ``norm'' of the four-velocity of a timelike particle is $-c^{2}$. An arbitrary choice of $\bar{p}_{1}, \bar{p}_{2}, \bar{p}_{3}$, determines $\bar{p}_{0}$ by,

\begin{equation}\label{a25}
\bar{p}_{0}=\frac{-\bar{g}^{0i}\bar{p}_{i}+\sqrt{(\bar{g}^{0j}\bar{g}^{0k}-\bar{g}^{00}\bar{g}^{jk})\bar{p}_{j}\bar{p}_{k}-\bar{g}^{00}m^{2}c^{2}}}{\bar{g}^{00}},
\end{equation}

\noindent so that Eq.\eqref{a24} holds.\\

\noindent Let $\varphi : \mathcal{M} \rightarrow \mathbb{R}$ by,

\begin{equation}\label{slice}
\varphi(x)=g(\exp_{\sigma(0)}^{-1}x,\, \sigma^{\prime}(0)).
\end{equation}

\noindent The space slice, $X$, at time coordinate $\tau=0=\bar{x}^{0}$, orthogonal to the observer's four-velocity (along $\sigma(\tau)$) is given by,

\begin{equation}\label{slice2}
X\equiv \varphi^{-1}(0).
\end{equation}

\noindent  The restriction $^{3}g$ of the space-time metric $g$ to $X$ makes $(X, ^{3}g)$ a Riemannian manifold, sometimes referred to as a Fermi surface or Landau submanifold, c.f., for example, \cite{bolos} and references therein.  

\begin{remark} The Fermi surface $X$ defines the collection of space-time points that are simultaneous with $\sigma(0)$ in a natural way relative to the observer $\sigma$. Hypersurfaces parallel to $X$ foliate a (possibly global) neighborhood of $\sigma$, and define a notion of simultaneity.  This foliation has been used to study geometrically defined relative velocities of distant objects and time dependent diameters of Robertson-Walker cosmologies \cite{Klein11, bolos2} 
\end{remark} 

\noindent Thus, using Eq.\eqref{a25}, the phase space for a single particle at time coordinate $\bar{x}^{0}=0$ is determined by the Fermi surface, $X$, together with the associated momentum coordinates, i.e., the cotangent bundle of the Fermi surface at fixed time coordinate $\bar{x}^{0}$. Let $\mathbb{P}$ be the sub bundle of the cotangent bundle of space-time with each  three dimensional fiber determined by (\ref{a25}). The proof of the following proposition is given in \cite{CK}.\\

\begin{proposition}\label{P1}
The volume 7-form 
$\tilde{\omega}$ on $\mathbb{P}$ given by, 
\begin{equation}
\tilde{\omega}=\frac{1}{\bar{p}^{0}}d\bar{x}^{0}\wedge d\bar{x}^{1}\wedge d\bar{x}^{2}\wedge d\bar{x}^{3}\wedge d\bar{p}_{1}\wedge d\bar{p}_{2}\wedge d\bar{p}_{3},\label{a26}
\end{equation}
is invariant  under all coordinate transformations.
\end{proposition}

\noindent As discussed in \cite{KC2}, phase space for a single particle is determined by
the Fermi surface, at fixed time coordinate, orthogonal to the observer's four-velocity (along $\sigma(\tau)$), along with the associated momentum
coordinates, i.e., the cotangent bundle of the Fermi surface at fixed time coordinate $\bar{x}^{0}$. The appropriate volume form is given by the interior product  
$\textbf{\textit{i}}(m\partial/\partial \tau)\tilde{\omega}$ of the four-momentum vector $m\partial/\partial\tau=\bar{p}^{\alpha}\partial/\partial\bar{x}^{\alpha}$ 
with $\tilde{\omega}$. Then,

\begin{equation}
m\textbf{\textit{i}}(\partial/\partial \tau)\tilde{\omega} = d\bar{x}^{1}\wedge
d\bar{x}^{2}\wedge d\bar{x}^{3}\wedge  d\bar{p}_{1}\wedge d\bar{p}_{2}\wedge d\bar{p}_{3}+(d\bar{x}^{0}\wedge\tilde{\psi}),\label{new27'}
\end{equation}

\noindent where $\tilde{\psi}$ is a five-form.  Since $d\bar{x}^{0}=0$ on vectors on phase space (with fixed time coordinate $\bar{x}^{0}$), the restriction of $m\textbf{\textit{i}}(\partial/\partial \tau)\tilde{\omega}$ to the one-particle (six-dimensional) phase space is,

\begin{equation}
m\textbf{\textit{i}}(\partial/\partial \tau)\tilde{\omega}=d\bar{x}^{1}\wedge d\bar{x}^{2}\wedge d\bar{x}^{3}\wedge  d\bar{p}_{1}\wedge d\bar{p}_{2}\wedge d\bar{p}_{3}.\label{new272}
\end{equation}

\noindent It follows from Proposition \ref{P1}
that the 6-form given by Eq.\eqref{new272} is invariant under coordinate changes of the space 
variables with $\bar{x}^{0}$ fixed. Thus,

\begin{equation}
\begin{split}\label{new27}
d\mathbf{\bar{x}}d\mathbf{\bar{p}}&\equiv
d\bar{x}^{1}\wedge d\bar{x}^{2}\wedge d\bar{x}^{3}\wedge  d\bar{p}_{1}\wedge d\bar{p}_{2}\wedge d\bar{p}_{3}\\
&=dx^{1}\wedge dx^{2}\wedge dx^{3}\wedge  dp_{1}\wedge dp_{2}\wedge dp_{3}\\
&\equiv dx^{1}  dx^{2}  dx^{3}   dp_{1}  dp_{2}  dp_{3}\\
&\equiv d\mathbf{x}d\mathbf{p}
\end{split}
\end{equation}

\noindent Physically this means that the
calculations that follow are independent of the choice of the orthonormal triad $e_{1}(0), e_{2}(0), e_{3}(0)$.\\

\noindent In what follows, we consider a particle system in contact with 
a heat bath, and consequently in thermal equilibrium. Following the usual convention, let $\beta=1/kT$, where $T$ is the temperature of the gas in some volume $\Lambda$, and $k$ is Boltzmann's constant.\\

\begin{remark}\label{temperature} An alternative convention for temperature may be used in what follows, consistent with an idea of Tolman's in the context of relativistic thermodynamics \cite{tolman}. At any location within the gas, one may identify an observer with four-velocity $\vec{K}/\|\vec{K}\|$. The energy of a particle with momentum $p$ measured by that observer is  $-K^{\alpha} p_{\alpha}/\|\vec{K}\|$.  Then, defining a position dependent equilibrium temperature $\tilde{T}$ by $\tilde{T}= T/\|\vec{K}\|= T/\sqrt{-\bar{g}_{00}}$ coincides with Tolman's formula and results in expressions equivalent to what we find in the sequel.
\end{remark}

\noindent It was shown in \cite{KC2} that the energy of a single particle may be naturally defined as $-K^{\alpha} p_{\alpha}-mc^{2}$, and that for a suitable volume $\Lambda$ (see below) the one particle partition function is given by,

\begin{equation}
\label{i1}
\tilde{Q}_{1}(\beta,\Lambda)\equiv\frac{1}{(2\pi\hbar)^{3}} \int_{ \Lambda}\int_{\mathbb{R}^{3}}
e^{\beta K^{\alpha}p_{\alpha}+\beta mc^{2}}\;d\mathbf{x}d\mathbf{p}=\frac{4\pi (mc)^{3}}{(2\pi\hbar)^{3}}e^{\beta mc^{2}} \int_{ \Lambda}\frac{K_{2}(\gamma)}{\gamma}d\mathbf{x},
\end{equation}

\noindent where $K_{2}(\gamma)$ is the modified Bessel function of the second kind evaluated at $\gamma$, and,

\begin{equation}\label{gamma}
\gamma\equiv\gamma(x^{1}, x^{2}, x^{3})\equiv\|\vec{K}\|\beta mc. 
\end{equation}

\begin{remark}\label{lie}
It is easy to verify that, 

\begin{equation}\label{Lie}
\mathcal{L}_{\partial/\partial\bar{x}^{0}}\left(e^{\beta(\bar{K}^{\alpha} \bar{p}_{\alpha}+mc^{2})}d\mathbf{\bar{x}}d\mathbf{\bar{p}}\right)=\mathcal{L}_{K}\left(e^{\beta(\bar{K}^{\alpha} \bar{p}_{\alpha}+mc^{2})}d\mathbf{\bar{x}}d\mathbf{\bar{p}}\right)=0,
\end{equation}

\noindent where $\mathcal{L}$ denotes Lie derivative. Analogous to the development of nonrelativistic statistical mechanics, this invariance together with Eq.\eqref{new27} is a partial justification for our choice of phase space measure.
\end{remark}

\noindent We take the phase space volume form for $n$ particles to be the product of $n$ copies of Eq.\eqref{new27}, one copy for each particle, which we denote as, 

\begin{equation}
\label{i1b}
d\mathbf{x}^{n}d\mathbf{p}^{n}=dx^{1}_{1}dx^{2}_{1}dx^{3}_{1}dp_{11}dp_{12}dp_{13}\ldots
dx^{1}_{n}dx^{2}_{n}dx^{3}_{n}dp_{n1}dp_{n2}dp_{n3}.
\end{equation}

\noindent It was shown in \cite{KC2} that this measure satisfies a Liouville theorem. We now assume that the $n$-particle \emph{equilibrium} distribution is determined by the following partition function:

\begin{equation}
\begin{split}\label{canonical}
\tilde{Q}_{n}(\beta,\Lambda, s)&\equiv\frac{e^{n\beta mc^{2}}}{(2\pi\hbar)^{3n}} \int_{\Lambda^{n}}\int_{\mathbb{R}^{3n}}e^{\beta (\sum_{i=1}^{n}K^{\alpha}p_{i\alpha}-\tilde{V}_{\Lambda}(\mathbf{x^{n}}|s))}d\mathbf{x}^{n}d\mathbf{p}^{n}\\
&=\left[\frac{4\pi (mc)^{3}e^{\beta mc^{2}}}{(2\pi\hbar)^{3}}\right]^{n} \int_{ \Lambda^{n}}e^{-\beta \tilde{V}_{\Lambda}(\mathbf{x^{n}}|s)}\left(\prod_{i=1}^{n}\frac{K_{2}(\gamma_{i})}{\gamma_{i}}\right)d\mathbf{x^{n}}
\end{split}
\end{equation}

\noindent where  $s$ denotes a configuration of particles outside of $\Lambda$ that influences the equilibrium configuration of particles $\mathbf{x^{n}}$ in $\Lambda$ ($s$ can be the empty configuration), $\tilde{V}_{\Lambda}(\mathbf{x^{n}}|s)\equiv \tilde{V}_{\Lambda}(x^{1}_{1},x^{2}_{1},x^{3}_{1},\ldots x^{1}_{n},x^{2}_{n},x^{3}_{n}|s)$ is a many body potential energy function whose form and restrictions are given in the following section, and $\gamma_{i}\equiv\gamma(x^{1}_{i},x^{2}_{i},x^{3}_{i})$. We note that this assumption excludes interaction functions of both position and momentum coordinates, a significant restriction.  However, a similar assumption for the case of special relativity was made in \cite{special}. We emphasize that we are not asserting that potential energy functions determine (via a Hamiltonian) the dynamics of general relativistic particle systems; instead, we assume only that equilibrium behavior is governed by expressions of the form of Eq.\eqref{canonical}.

\section*{\normalsize{3. Grand canonical ensemble on the Fermi space slice}}

\noindent Because the integrals in Eq.\eqref{canonical} factor as a product of configuration integrals and momentum integrals (which can be evaluated explicity), we focus on the former.  To proceed further, it is convenient to work with the grand canonical ensemble.   A grand canonical formalism for particle systems on Riemannian manifolds is described in \cite{rockner, albeverio, kuna} and references therein, briefly summarized here in a way that is useful for our purposes.\\  

\noindent Let $\mathcal{B}(X)$ denote the $\sigma$-algebra of Borel sets on $X$ and $\mathcal{B}_{c}(X)$ represent the collection of sets in $\mathcal{B}(X)$ with compact closure. Let $\Gamma_{X}$ be the set of particle configurations on $X$, i.e.,

\begin{equation}\label{configurations}
\Gamma_{X}= \{x \subset X: |x
\cap K|<\infty \,\,\text{for any compact}\,\,K \subset X\},
\end{equation}

\noindent where $|A|$ denotes cardinality of the set $A$. For  $\Lambda\in\mathcal{B}_{c}(X)$, define $N_{\Lambda}: \Gamma_{X}\rightarrow\mathbb{N}_{0}$, the set of natural numbers (including zero) by,

\begin{equation}\label{rv}
N_{\Lambda}(x) = |x\cap\Lambda|.
\end{equation}

\noindent Let $\mathcal{B}(\Gamma_{X})$ be the $\sigma$-algebra generated by all such functions.  For $\Lambda\subset X$, let $x_{\Lambda}=x\cap \Lambda$ and,

\begin{equation}\label{configurations2}
\Gamma_{\Lambda}= \{x \in \Gamma_{X}: x_
{X/\Lambda}=\varnothing\}.
\end{equation}

\noindent For $n=0,1,2,...$, let,

\begin{equation}\label{configurations3}
\Gamma_{\Lambda}^{(n)}= \{x \in \Gamma_{\Lambda}: |x|=n\}.
\end{equation}

\noindent There is a natural bijection,

\begin{equation}\label{bijection}
\tilde{\Lambda}^{n}/S_{n}\rightarrow\Gamma_{\Lambda}^{(n)},
\end{equation}

\noindent where $S_{n}$ is the permutation group over $\{1,...,n\}$ and,

\begin{equation}\label{tilde}
\tilde{\Lambda}^{n} \equiv \{(x_{1}, ..., x_{n}) \in \Lambda^{n} : x_{i} \neq x_{j}\, \,\text{if}\,\, i\neq j\}.
\end{equation}

\noindent If $\Lambda\in\mathcal{B}_{c}(X)$ is open, this correspondence determines a locally compact, metrizable topology on $\Gamma_{\Lambda}^{(n)}$ and then, 

\begin{equation}
\Gamma_{\Lambda}=\bigcup_{n=0}^{\infty}\Gamma_{\Lambda}^{(n)}
\end{equation}

\noindent is equipped with the topology of disjoint union and hence a Borel $\sigma$-algebra, $\mathcal{B}(\Gamma_{\Lambda})$.  Then $(\Gamma_{X}, \mathcal{B}(\Gamma_{X}))$ is the projective limit of the measurable spaces $(\Gamma_{\Lambda}, \mathcal{B}(\Gamma_{\Lambda}))$ as $\Lambda$ increases to $X$, \cite{rockner}.\\ 

\noindent For $\Lambda\in\mathcal{B}_{c}(X)$, let $\mathcal{B}_{\Lambda}=\sigma\{N_{\Lambda'}: \Lambda'\subset\Lambda,\,\,\Lambda'\in\mathcal{B}_{c}(X)\}$. The $\sigma$-algebras $\mathcal{B}(\Gamma_{\Lambda})$ and $\mathcal{B}_{\Lambda}$ are $\sigma$-isomorphic.

\begin{definition}\label{cylinder}
Let $\Lambda\in\mathcal{B}_{c}(X)$.  A set $A\in\mathcal{B}_{\Lambda}\subset\mathcal{B}(\Gamma_{X})$ is called a cylinder set.
\end{definition}

\begin{remark}\label{cylinderfunction}
If $A$ is a cylinder set, let  $A_{\Lambda}  =\{x_{\Lambda}:x \in A \}$. Then $1_A(x)=1_{A_{\Lambda}}(x_{\Lambda})$. To see this,  let $\mathcal {C}=\{ A \subset \Gamma _X: 1_A(x)=1_A(x_{\Lambda} \cup s_{X/\Lambda}) \mbox{ for all } s \in \Gamma _X \} $.
 Then $1_A(x)=1_{A_{\Lambda}}(x_{\Lambda})$ holds if and only if $A \in \mathcal {C}$. Let $\mathcal {D}$ be the class of the sets of the form $\{x: N_{\Lambda _1}(x)=n_1, N_{\Lambda _2}(x)=n_2,...,N_{\Lambda _k}(x)=n_k \} $, for some nonnegative integers $k, n_1, n_2, ..., n_k $, and $\Lambda _1, \Lambda _2, ..., \Lambda _k \subset \Lambda$.  Then $\mathcal {D}$ is closed under intersection and $\mathcal {D} \subset \mathcal {C}$.  Moreover, $\mathcal {C}$ has the following properties (i) $ \Gamma _X \in \mathcal {C}$, (ii) if $A, B \in \mathcal {C}$ and $A \subset B$ then $B/A \in \mathcal {C}$ and (iii) if $A_n \in \mathcal {C}$ and $A_n \uparrow A$, then $A \in \mathcal {C}$.  By Dynkin's Monotone Class Theorem, $\mathcal {C}$ contains  $\sigma (\mathcal {D})=\mathcal{B}_{\Lambda}$.

\end{remark}

\noindent We next define a Poisson measure on $X$, based on the measure $d\mathbf{x}$, determined by Eq.\eqref{new27}. Let $\Lambda\in\mathcal{B}_{c}(X)$ be open and let $s_{\Lambda}^{n}:\tilde{\Lambda}^{n}\rightarrow\Gamma_{\Lambda}^{(n)}$ by $s_{\Lambda}^{n}((x_{1}, ..., x_{n}))=\{x_{1}, ..., x_{n}\}$.  Then $d\mathbf{x}^{n}\circ (s_{\Lambda}^{n})^{-1}$ is a measure on $\Gamma_{\Lambda}^{(n)}$ which we again denote by $d\mathbf{x}^{n}$ (distinguished from Eq.\eqref{i1b} by context). For the case $n=0$, $d\mathbf{x}^{0}(\varnothing)\equiv 1$.  For nonnegative activity $z$, and $\Lambda\in\mathcal{B}_{c}(X)$, we can define (un normalized) Poisson measure on $\Gamma_{\Lambda}$ by,

\begin{equation} \label{poisson1}
\nu_{\Lambda} (d\mathbf{x})= \sum_{n=0}^{\infty} \frac{z^{n}}{n!} d\mathbf{x}^{n}.
\end{equation}

\noindent For each $\Lambda\in\mathcal{B}_{c}(X)$, let $V_{\Lambda}$ be a $\mathcal{B}(\Gamma_{\Lambda})$ measurable interaction potential of the form,

\begin{equation}\label{V}
V_{\Lambda}(x)=\sum_{n=0}^{|x|}\, \sum _{y \subset x, |y|=n}\phi_{n}(y),
\end{equation}

\noindent for finite configuration $x$, where $\phi_{n}: \Gamma_{\Lambda}^{(n)}\rightarrow\mathbb{R}\cup\{+\infty\}$, $\inf_{n,y}\phi(y)>-\infty$. For the case $n=0$, we define $\phi_{0}(\varnothing)=V_{\Lambda}(\varnothing)=|\Lambda|\rho_{\text{vac}}$, where $|\Lambda|$ is the volume of $\Lambda$  determined by the Riemannian metric $^{3}g$ on $X$ and $\rho_{\text{vac}}$ is the energy density of the vacuum given by

\begin{equation}\label{vacuumenergy}
\rho_{\text{vac}} =\frac{\lambda c^{4}}{8\pi G},
\end{equation}

\noindent where $\lambda$ is the cosmological constant (possibly zero) and $G$ is Newton's universal gravitational constant.

\begin{definition}\label{finiterange}
An interaction potential of the form given by Eq.\eqref{V} has finite range if for each  $n\geq2$, there exists $R>0$ such that $\phi_{n}(x_{1}, ..., x_{n})= 0$ whenever the proper distance (as determined by the Riemannian metric $^{3}g$) from $x_{i}$ to $x_{j}$ exceeds $R$ for some $i$ and $j$. 
\end{definition}

\noindent Though not essential, we assume henceforth that each $\phi_{n}$ has finite range  for $n\geqslant3$, so that only the two-body interactions might have infinite range.  In order to develop notational consistency with Eq.\eqref{canonical}, we also define a temperature dependent, one-body potential (or equivalently, a temperature independent external field) by,

\begin{equation}\label{1body}
e^{- \beta \phi _1 (x^1, x^2, x^3)}=\frac {K_2(\gamma (x^1, x^2, x^3))}{\gamma (x^1, x^2, x^3)}.
\end{equation}

\noindent For configurations $x\subset \Lambda$ and $s \subset X$ define,

\begin{equation}\label{V2}
V_{\Lambda}(x|s)=|\Lambda|\rho_{\text{vac}}+\sum_{n=1}^{\infty}\, \sum _{\substack{y \subset x\cup (s\cap X/\Lambda)\\
y\cap x\neq\varnothing\\|y|=n}}
\phi_{n}(y),
\end{equation}

\noindent if the right hand side converges absolutely, and $+\infty$ otherwise. For $\Lambda\in\mathcal{B}_{c}(X)$ and a boundary configuration $s$, define the partition function by,

\begin{equation}\label{partition}
Z_{\Lambda}(s)=Z_{\Lambda}(s, \beta, z)= \int_{\Gamma_{\Lambda}}e^{-\beta V_{\Lambda}(x|s)}\nu_{\Lambda} (dx),
\end{equation}

\noindent where we assume that a factor $4\pi (mc)^{3}e^{\beta mc^{2}}/(2\pi\hbar)^{3}$ is included in the activity $z$ through Eq.\eqref{poisson1} so that,

\begin{equation}\label{activity}
z= \frac{4\pi (mc)^{3}}{(2\pi\hbar)^{3}}e^{\beta (mc^{2}+\mu)},
\end{equation}

\noindent where the parameter $\mu$ is chemical potential. With this identification, Eq.\eqref{partition} may be rewritten as,

\begin{equation}
Z_{\Lambda}(s, \beta, z)= \sum_{n=0}^{\infty}\frac{z^{n}}{n!}Q_{n}(\beta,\Lambda, s),
\end{equation}

\noindent where $Q_{0}(\beta,\Lambda, s)=\exp (-\beta |\Lambda |\rho_{\text{vac}})$ and for $n\geq 1$,

\begin{equation}\label{Q}
Q_{n}(\beta,\Lambda, s)=\int_{\Lambda^{n}}e^{-\beta V_{\Lambda}(x|s)}d\mathbf{x^{n}}.
\end{equation}

\noindent Following the grand canonical statistical mechanical prescription, the connection to thermodynamics is given by the following definitions.

\begin{definition}\label{finitepressure} The finite volume pressure $P_{\Lambda}(\beta, z)$ for $\Lambda\in\mathcal{B}_{c}(X)$ is given by,

\begin{equation}\label{pressure}
\beta P_{\Lambda}(\beta, z)=\frac{\log Z_{\Lambda}(\varnothing)}{|\Lambda|},
\end{equation}
\noindent where the volume $|\Lambda|$ determined by the Riemannian metric $^{3}g$ on $X$.
\end{definition}

\begin{definition}\label{defpressure} For the case that $X$ is unbounded, we define the infinite volume pressure $P=P(\beta, z)$ by,

\begin{equation}\label{pressure}
\beta P(\beta, z)=\lim_{n\rightarrow\infty}\frac{\log Z_{\Lambda_{n}}(\varnothing)}{|\Lambda_{n}|},
\end{equation}
\noindent provided the limit exists. Here $\Lambda_{n}$ is the ball of radius $n$ centered at the origin of Fermi-Walker-Killing coordinates.
\end{definition}

\noindent Following \cite{kuna}, Eqs.\eqref{poisson1} -- \eqref{partition} determine a local specification $\Pi$.  For any $s\in\Gamma_{X}$, $\Lambda\in\mathcal{B}_{c}(X)$, and $A\in\mathcal{B}(X)$, let, 

\begin{equation}\label{specification}
\Pi_{\Lambda}(A,s)=\frac{1_{Z_{\Lambda}(s)<\infty}(s)}{Z_{\Lambda}(s)}\int_{\Gamma_{\Lambda}}1_{A}(x\cup s_{X/\Lambda})e^{-\beta V_{\Lambda}(x|s)}\nu_{\Lambda} (dx).
\end{equation}

\noindent We define a Gibbs state via the DLR equations \cite{kuna}.

\begin{definition}\label{gibbs}
Assume that $X$ is unbounded. A probability measure $\mu$ on $(\Gamma_{X}, \mathcal{B}(\Gamma_{X}))$ is a Gibbs state if
$$\int_{\Gamma_{X}}\Pi_{\Lambda}(A,s)\mu(ds)=\mu(A)$$
\noindent for all $A\in\mathcal{B}(\Gamma_{X})$ and $\Lambda\in\mathcal{B}_{c}(X)$. 
\end{definition}

\begin{remark}\label{existence}
For the case of pair potentials, conditions have been given for the existence of Gibbs states in \cite{kuna}; see also \cite{kondratiev}.  Note also that the specification given by Eq.\eqref{specification}, and hence any Gibbs state, is invariant under changes of nonzero energy density of the vacuum $\phi_{0}$ (see Eq.\eqref{vacuumenergy}).  However, the pressure is not. 
\end{remark}

\section*{\normalsize{4. Ideal Gas Laws}}

\noindent  In the context of general relativity, we define an ideal gas to be the ensemble determined by a potential $V$ satisfying $\phi_{n} \equiv 0$ for all $n\geq2$. Note, however, that $\phi_{1}\neq0$ (see Eq.\eqref{1body}), and $\phi_{0}$, which is proportional to the vacuum energy given by Eq.\eqref{vacuumenergy}, is not necessarily zero. Using the definitions of the preceding sections, a derivation of a finite volume general relativistic ideal gas law for an observer following a path tangent to a timelike Killing vector, using the grandcanoncial ensemble is straightforward.  By Def. \ref{finitepressure},

\begin{equation}\label{idealgas1}
P_{\Lambda}(\beta, z)= -\rho_{\text{vac}}+\frac{z}{\beta|\Lambda|}\int_{ \Lambda}\frac{K_{2}(\gamma(x))}{\gamma(x)}d\mathbf{x}.
\end{equation}

\noindent The expected number of particles, $\langle N_{\Lambda} \rangle$, in $\Lambda$ is given by,

\begin{equation}\label{idealgas2}
\langle N_{\Lambda} \rangle= \frac{1}{Z_{\Lambda}(\varnothing)} \sum_{n=0}^{\infty}n\frac{z^{n}}{n!}Q_{n}(\beta,\Lambda, \varnothing)=z\frac{\partial}{\partial z} \log Z_{\Lambda}(\varnothing)=z\int_{ \Lambda}\frac{K_{2}(\gamma(x))}{\gamma(x)}d\mathbf{x}.
\end{equation}

\noindent Combining Eqs. \eqref{idealgas1} and \eqref{idealgas2} then yields,

\begin{equation}\label{idealgas3}
P_{\Lambda}= -\rho_{\text{vac}}+\frac{\langle N_{\Lambda} \rangle}{|\Lambda|}kT = -\rho_{\text{vac}}+\rho_{\Lambda}kT,
\end{equation}

\noindent where $\rho_{\Lambda}$ is the number density of particles in $\Lambda$.  In the case that $\rho_{\text{vac}}=0$, this is the standard form of the ideal gas law, with the volume $|\Lambda|$ determined by the induced metric $^{3}g$ on the Fermi surface for the observer.  The result is well-known for the case of Minkowski space-time.  As in the classical (Newtonian) case, it is easily shown (using Eq.\eqref{idealgas2}) that the particle number is Poisson distributed.  Thus, if $\text{Pr}_{\Lambda}(N)$ represents the probability that $N$ particles are in $\Lambda$, then,

\begin{equation}\label{Poisson}
\text{Pr}_{\Lambda}(N)= e^{-\langle N_{\Lambda} \rangle}\frac{\,\,\langle N_{\Lambda} \rangle^{N}}{N!}.
\end{equation}

\noindent It follows from Eqs.\eqref{idealgas2}, \eqref{idealgas3}, and \eqref{Poisson} that the thermodynamic behavior of an ideal gas, and in particular the particle density, is determined by the metric of the space-time through the norm of the timelike Killing field, $||\vec{K}||$ (see Eq.\eqref{gamma}), and the volume $|\Lambda|$ (determined by $^{3}g$). Thus in principle, the geometry of space-time may be studied through the behavior of an ideal gas within it. \\

\noindent By way of illustration, we compare the thermodynamic behavior of an ideal gas in de Sitter space and the Einstein static universe. It was shown in \cite{KC3} that the metric for the Einstein static universes in polar form using Fermi coordinates for any geodesic observer (with path $\sigma(t)=(t,0,0,0)$ in these coordinates) is given by,

\begin{equation}\label{Einsteinpolar}
\text{Einstein static:}\,\,\,\,ds^{2}=-c^{2}dt^{2}+d\rho^{2}+R^{2}\sin^{2}\left(\frac{\rho}{R}\right)(d\theta^{2} + \sin^{2} \theta \,d\phi^{2}),
\end{equation}

\noindent where $\rho< \pi R$ and $R = 1/\sqrt{\lambda}$ is the radius of the universe.  It was also shown in \cite{KC3} that the metric for de Sitter space-time, in polar form of Fermi coordinates for any geodesic observer is given by,

\begin{equation}\label{deSitterpolar}
\text{de Sitter:}\,\,\,\,ds^{2}= -c^{2}\cos^{2}\left( a\rho \right) dt^{2}+d\rho^{2}+\frac{\sin^{2}(a\rho)}{a^{2}}(d\theta^{2} + \sin^{2} \theta \,d\phi^{2}),
\end{equation}

\noindent where $a =\sqrt{ \lambda/3}$ and $\rho<\pi/2a$ so that the coordinates cover the space-time up to the cosmological horizon of the observer. For both space-times, $\vec{K}=\partial/\partial t$ is a timelike Killing vector field tangent to the respective timelike geodesics of the observers.  Then from Eq.\eqref{gamma}, we find for these space-times,

\begin{equation}\label{gamma2}
\text{Einstein static:}\,\,\,\,\gamma(x^{1}, x^{2}, x^{3})=\beta mc^{2}, 
\end{equation}

\noindent as distinguished from, 

\begin{equation}\label{gamma3}
\text{de Sitter:}\,\,\,\,\gamma(x^{1}, x^{2}, x^{3})=\beta mc^{2}\cos\left( a\rho \right), 
\end{equation}\

\noindent where in terms of Cartesian Fermi coordinates, $\rho=\sqrt{(x^{1})^{2}+(x^{2})^{2}+(x^{3})^{2}}$.  The expected number of particles for each space-time is determined by Eq.\eqref{idealgas2}, and respective volumes are calculated from Eqs.\eqref{Einsteinpolar} and \eqref{deSitterpolar} after setting $dt=0$. For volumes determined by a given set of Fermi coordinates for the two space-times, Eq.\eqref{idealgas3} and \eqref{vacuumenergy} yield,

\begin{equation}\label{idealgas4}
P_{\Lambda} = P= -\frac{\lambda c^{4}}{8\pi G}+\rho_{\text{mass}}\frac{kT}{m},
\end{equation}

\noindent where $\rho_{\text{mass}} = m\rho_{\Lambda}$ is the mass density of the ideal gas when $m$ is the mass of each gas particle, and these densities for the two space-times depend on the respective values of $\gamma$, and will in general be different.  So the relativistic thermodynamic behavior of an ideal gas distinguishes the two space-times for any common value of $\lambda$.\\

\begin{remark}\label{static} It is interesting to note that the equation of state for the perfect fluid of the Einstein static universe is $P=0$, with zero temperature, i.e., the fluid is ``dust,'' and the mass density $\rho$  determined by the Einstein field equations is,

\begin{equation}\label{massdensity}
\rho =\frac{\lambda c^{2}}{4\pi G}.
\end{equation}

\noindent If we set $\rho_{\text{mass}}$ for the ideal gas of test particles equal to $\rho$ of the perfect fluid and require $P=0$, the temperature is given by,

\begin{equation}\label{temp}
kT = \frac{mc^{2}}{2},
\end{equation}

\noindent and approaches zero only as $m\rightarrow0$. Thus, the thermodynamic behavior of the ideal gas models that of the perfect fluid (i.e. dust) only in the limiting case that mass of each test particle of the gas goes to zero.\\
\end{remark}

\section*{\normalsize{5. Newtonian limits}}

\noindent For a timelike path $\sigma$ tangent to a Killing vector $\vec{K}$ in a given space-time, we have defined through Eqs.\eqref{idealgas1} -- \eqref{Poisson} the relativistic grand canonical pressure  of an ideal gas.  Under appropriate conditions, there is associated with such a statistical mechanical system an analogous Newtonian (i.e., nonrelativistic) expression for the pressure of the gas.  It should be noted that the term ``ideal gas'' is somewhat misleading in the context of general relativity. This is because a volume of gas subject to no forces is still affected by the curvature of space-time, and this corresponds to a Newtonian gas subject to gravitational, tidal, and in some instances ``centrifugal forces,'' depending on the path of the observer.  This associated Newtonian partition function and finite volume Gibbs state for an ideal gas is obtained as a limit as $c \rightarrow \infty$.\\

\begin{definition}\label{alpha}
For a timelike Killing field $\vec{K}$ tangent to the path $\sigma$ parameterized by proper time $\tau$, define a dimensionless function of the space coordinates on the Fermi surface, $\alpha= \alpha (x^{1}, x^{2}, x^{3})$,  by $\|\vec{K}\|=\alpha c$.
\end{definition}

\noindent Throughout this section we let, $y= 1/c$ so that the Newtonian limit is obtained as $y \rightarrow 0^{+}$, and we assume that $\alpha$ can be extended to a smooth function of $y, x^{1}, x^{2}, x^{3}$ including at $y =0$ (here and below we suppress the dependence of $\alpha$ on $x^{1}, x^{2}, x^{3}$).  We refer to this function as $\alpha (y)$ and further assume that $\alpha(0)=1$ and $\alpha'(0)=0$ (see the following section).  Then,

\begin{equation} \label{alpha2}
\alpha(y)= 1+\frac{ \alpha''(0)}{2} y^{2} + \frac{\alpha'''(y_{0})}{6} y^{3}= 1+\frac{1}{2} \alpha''(0) y^{2} + O(y^{3}),
\end{equation}

\noindent where $0<y_{0}<y$. Note that $\alpha''(0)$ is a function of $x^{1}, x^{2},x^{3}$. Below we will identify $\frac{1}{2}m\alpha''(0)$ as the Newtonian, nonrelativistic potential energy $U(x^{1}, x^{2},x^{3})$ of a test particle of mass $m$ with coordinates $(0,x^{1}, x^{2},x^{3})$ due to the gravitational field. This potential energy function $U(x^{1},x^{2},x^{3})$ is normalized (by an additive constant) so that $U(0,0,0) = 0$, i.e., the potential energy on the timelike path $\sigma$ is zero.  This is because $\vec{K}$ is the four velocity on $\sigma$ so $\alpha(y) =\|\vec{K}\|/c$ is identically $1$ as a function of $y=1/c$ on  $\sigma$, forcing $\alpha''(0) = 0$ there.

\begin{definition} \label{Newton}  Let $\Lambda\in\mathcal{B}_{c}(X)$. For a particle system following a timelike path $\sigma$, we define the following Newtonian statistical mechanical analogs.

\begin{enumerate}
\item[(a)] The Newtonian activity paramater, $z_{\text{Newt}}$ is given by,

\begin{equation}\label{activityN}
z_{\text{Newt}}=\left(\frac{m}{2\pi\hbar^{2}\beta}\right)^{\frac{3}{2}}e^{\beta\mu},
\end{equation}
\noindent where $\mu$ is the chemical potential (as in Eq.\eqref{activity}). 

\item [(b)]   $Z_{\text{Newt}}(\Lambda)$ denotes the partition function for a non relativistic ideal gas subject to an external field with potential energy function $\frac{1}{2} m\alpha''(0)$ and activity $z_{\text{Newt}}$ so that,

\begin{equation}
\begin{split}\label{znewton}
Z_{\text{Newt}}(\Lambda)=& \sum_{n=0}^{\infty}\frac{z_{\text{Newt}}^{n}}{n!}\left[\int_{\Lambda}e^{-\frac{1}{2}\beta m\alpha''(0)}\;d\mathbf{x}\right]^{n}\\
=&\exp\left[ z_{\text{Newt}}\int_{\Lambda}e^{-\frac{1}{2}\beta m\alpha''(0)}\;d\mathbf{x}\right]
\end{split}
\end{equation}
\item[(c)] Let $\Pi_{\Lambda}^{\text Newt}(A,s)$ denote the specification associated with the one-body Newtonian potential $\frac{1}{2}\beta m\alpha''(0)$ and partition function $Z_{\text{Newt}}(\Lambda)$ given by Eq.\eqref{specification}.
\end{enumerate}

\end{definition}

\begin{theorem}\label{Newtonianlimit} 
Let $\sigma$ be a timelike path with $\alpha$ satisfying Eq.\eqref{alpha2}. Assume that $\Lambda\in\mathcal{B}_{c}(X)$. Then
\begin{enumerate}
\item [(a)] If $Z_{\Lambda}$ is the partition function for an ideal gas for an observer in a space-time with $\rho_{\text{vac}}=0$ (i.e. with zero cosmological constant), then,

\begin{equation}\label{partitionlimit}
 \lim_{\,\,c \rightarrow\infty}Z_{\Lambda}=Z_{\text{Newt}}(\Lambda),
\end{equation}
\noindent where as above $c$ is the speed of light.
\item [(b)] If $A\in\mathcal{B}_{\Lambda}$ then,

\begin{equation}\label{specificationlimit}
 \lim_{\,\,c \rightarrow\infty}\Pi_{\Lambda}(A,s)=\Pi_{\Lambda}^{\text Newt}(A,s)
\end{equation}
\noindent (where both sides of the equation are independent of $s$).
\end{enumerate}

\end{theorem}

\begin{proof} a) The asymptotic behavior of $K_{2}(\gamma)$ for large argument is given by (see \cite{E53}),

 \begin{equation}\label{asymp}
e^{- \beta \phi _{1}}=\frac{K_2(\gamma)}{\gamma} \sim \sqrt{\frac{\pi}{2 \gamma^{3}}} e^{-\gamma}.
\end{equation}

\noindent Since $\vec{K}$ is timelike, $\|\vec{K}\|$ is nonzero and  bounded below on $\Lambda$. It follows that $\gamma=\|\vec{K}\|\beta mc\rightarrow\infty$ uniformly on $\Lambda$ as $c\rightarrow\infty$.  By \eqref{asymp}, given any $\epsilon>0$, there exists $\delta>0$ such that $1/c=y<\delta$ implies,

 \begin{equation}\label{asymp2}
\left|\frac{e^{-\beta \phi _{1}}}{\sqrt{\frac{\pi}{2 \gamma^{3}}} e^{-\gamma}}-1\right|<\epsilon.
\end{equation}
Multiplying by $z$, substituting Eqs.\eqref{activity}, \eqref{activityN},  and $\gamma=\alpha\beta mc^2$, and rearranging terms yields,

 \begin{equation}\label{asymp3}
\left|ze^{- \beta \phi _{1}}- z_{\text{Newt}}\frac{e^{\beta mc^{2}(1-\alpha)}}{\alpha^{3/2}}\right|<\epsilon z_{\text{Newt}}\frac{e^{\beta mc^{2}(1-\alpha)}}{\alpha^{3/2}}.
\end{equation}
Since $\alpha$ is smooth and the closure of $\Lambda$ is compact, it follows from Eq.\eqref{alpha2} that $\alpha\rightarrow1$ and $\beta mc^{2}(1-\alpha)\rightarrow-\beta m\alpha''(0)/2$, both uniformly on $\Lambda$ as $y\rightarrow0$.  It then follows from Eq.\eqref{asymp3} and the triangle inequality that,

\begin{equation}\label{uniform}
ze^{- \beta \phi _{1}}\rightarrow z_{\text{Newt}}e^{-\frac{1}{2}\beta m\alpha''(0)},
\end{equation}
uniformly on $\Lambda$ as $y\rightarrow0$. Integrating both sides of Eq.\eqref{uniform} then gives, 

\begin{equation}\label{i14}
\lim_{\,\,c \rightarrow\infty}zQ_{1}=z_{\text{Newt}}\int_{\Lambda}e^{-\frac{1}{2}\beta m\alpha''(0)}\;d\mathbf{x}.
\end{equation}

\noindent Thus, 

\begin{equation}\label{theorem1}
 \lim_{\,\,c \rightarrow\infty}Z_{\Lambda}=\lim_{\,\,c \rightarrow\infty}\exp[zQ_{1}]=\exp\left[ z_{\text{Newt}}\int_{\Lambda}e^{-\frac{1}{2}\beta m\alpha''(0)}\;d\mathbf{x}\right]=Z_{\text{Newt}}(\Lambda).
\end{equation}

\noindent
\noindent b) Let $A\in\mathcal{B}_{\Lambda}$ be a cylinder set, and let $A_{\Lambda}  =\{x_{\Lambda}:x \in A \}$.  By Remark \ref{cylinderfunction},
\begin{eqnarray}
1_A(x)=1_{A_{\Lambda}}(x_{\Lambda}). \label{cylinder set equation}
\end{eqnarray}
Since $A_{\Lambda} \subset \Gamma _{\Lambda}$ we may write $A_{\Lambda} =\cup_{n=0}^\infty A_{\Lambda}^n $, where 
$A_{\Lambda}^n=A_{\Lambda} \cap \Gamma _{\Lambda}^{(n)}$.  From Eqs. \eqref{poisson1} and \eqref{specification}, 

\begin{equation}
 \Pi _{\Lambda } (1_A, s)=\frac{1}{Z_{\Lambda}}\sum _{n=0}^{\infty} \frac{1}{n!} \int_{\Lambda^{n} }1_{{A}_{\Lambda}^n}(x_1, ..., x_n)\prod_{i=1}^{n}\left(z e^{- \beta \phi _{1}(x_{i})}\right) d\mathbf{x}^{n}.
\end{equation}

\noindent By Eq.\eqref{uniform}, there exists $c_0$ independent of $(x^{1}, x^{2}, x^{3})$ such that

\begin{equation}\label{uniform2}
ze^{- \beta \phi _{1}}\leq z_{\text{Newt}}e^{-\frac{1}{2}\beta m\alpha''(0)}+1,
\end{equation}
for all  $c \ge c_0$ and all $(x^{1}, x^{2}, x^{3}) \in \Lambda$. Since $\alpha''(0)$ is bounded on $\Lambda$, there is a constant $C$ such that

\begin{equation}\label{uniform3}
\begin{split}
\sum _{n=0}^{\infty} \frac{1}{n!}& \int_{\Lambda^{n} }1_{{A}_{\Lambda}^n}(x_1, ..., x_n)\prod_{i=1}^{n}\left( z_{\text{Newt}}e^{-\frac{1}{2}\beta m\alpha''(0)}+1\right) d\mathbf{x}^{n}\\
\le \sum _{n=0}^{\infty} &\frac{1}{n!} \int_{\Lambda^{n} }1_{{A}_{\Lambda}^n}(x_1, ..., x_n)C^n d\mathbf{x}^{n}
\le e^{C \int_{\Lambda}1 d\mathbf{x}} < \infty.
\end{split}
\end{equation}
Therefore part (b) follows from the Dominated Convergence Theorem and Eq.\eqref{theorem1}.

\end{proof}

\begin{corollary}\label{idealgibbs}
Assume that the Fermi surface $X$ is unbounded and let $\mu$ be the unique Gibbs state for the ideal gas.  Let $\mu_{\text Newt}$ be the unique Newtonian Gibbs state for the specification $\{\Pi_{\Lambda}^{\text Newt}\}$.  Then for any cylinder set $A$,

\begin{equation}\label{gibbslimit}
\lim_{\,\,c \rightarrow\infty}\mu(A) = \mu_{\text Newt}(A).
\end{equation}

\end{corollary}
\begin{proof}
Existence and uniqueness of the Gibbs state for an ideal gas follows from Kolmogorov's theorem for projective limit spaces (see e.g. \cite{kuna} p. 21 or \cite{prob} Chap V, Theorem 3.2).  Since $A$ is a cylinder set, there exists $\Lambda\in\mathcal{B}_{c}(X)$ such that $\Pi_{\Lambda}(A,s)$ is independent of $s$.  It follows from Eq.\eqref{gibbs} and Theorem \ref{Newtonianlimit} that,

\begin{equation}\label{gibbslimit2}
\lim_{\,\,c \rightarrow\infty}\mu(A) = \lim_{\,\,c \rightarrow\infty}\Pi_{\Lambda}(A,s)=  \lim_{\,\,c \rightarrow\infty}\Pi_{\Lambda}^{\text Newt}(A,s)=\mu_{\text Newt}(A).
\end{equation}

\end{proof}

\begin{definition}\label{Newtonpressure}  Let $P_{\Lambda}(\beta, z)$, given by Eq.\eqref{idealgas1}, be the pressure of an ideal gas according to a timelike observer satisfying the hypotheses of Theorem \ref{Newtonianlimit}a. We define the associated Newtonian pressure $P_{\text{Newt}}(\beta, z, \Lambda)$ to be, 

\begin{equation}\label{znewton2}
\beta P_{\text{Newt}}(\beta, z, \Lambda)= \left(\int _{\Lambda}1\, d\mathbf{x}\right)^{-1}\log Z_{\text{Newt}}(\Lambda).
\end{equation}

\noindent In the case that $\rho_{\text{vac}}\neq0$, we define the associated Newtonian pressure by,

\begin{equation}\label{znewton3}
\beta P_{\text{Newt}}(\beta, z, \Lambda)=-\rho_{\text{vac}}+ \left(\int _{\Lambda}1\, d\mathbf{x}\right)^{-1}\log Z_{\text{Newt}}(\Lambda).
\end{equation}

\end{definition}

\noindent We note that in the preceding definition, the role of volume $\Lambda$ is restricted solely to the identification of (constant) limits of integration for the volume integral.  In the relativistic context, these limits of integration are determined by proper lengths of the dimensions of the container of gas, and in the Newtonian limit the limits of integration are absolute length measurements.\\

\noindent Our inclusion of the vacuum energy density term in Eq.\eqref{znewton3} is motivated in part by generalized Buchdahl inequalities that establish a minimum density of matter in space-times with a positive cosmological constant \cite{boehmer}.  Evidently, a collection of particles falling below a certain critical density will be pushed apart by a positive vacuum energy.

\section*{\normalsize{6. Kerr space-time}}

\noindent In this section we calculate $\alpha(y)$ and the limiting Newtonian potential energy function, $U(x^{1}, x^{2},x^{3})\equiv m\alpha''(0)/2$ for the case of circular geodesic orbit in the equatorial plane of Kerr space-time. The Kerr metric in Boyer-Lindquist coordinates is given by (recall that $y=1/c$),

\begin{equation}
\begin{split}\label{Kerr1}
ds^{2}=&-\left(1-\frac{2y^{2}GMr}{\rho^{2}}\right)\frac{1}{y^{2}}dt^{2}-\frac{4y^{2}GMar\sin^{2}\theta}{\rho^{2}}dtd\phi+\frac{\Sigma}{\rho^{2}}\sin^{2}\theta d\phi^{2}\\
&+\frac{\rho^{2}}{\Delta}dr^{2}+\rho^{2}d\theta^{2},
\end{split}
\end{equation}

\noindent where,
\begin{equation}
\begin{split}\label{Kerr2}
\rho^{2}&=r^{2}+y^{2}a^{2}\cos^{2}\theta,\\
\Delta&=r^{2}-2GMy^{2}r+y^{2}a^{2},\\
\Sigma&=\left(r^{2}+y^{2}a^{2}\right)^{2}-y^{2}a^{2}\Delta\sin^{2}\theta,\\
\end{split}
\end{equation}

\noindent and where $G$ is the gravitational constant, $M$ is mass, and $a$ is the angular momentum per unit mass, and $-GMy\leqslant a \leqslant GMy$.\\ 

\noindent Below we will need  notation for $\Delta$, and $\Sigma$, evaluated at  the specific coordinates $r=r_{0}$, and $\theta=\pi/2$.  For that purpose, we define,
\begin{equation}
\begin{split}\label{Kerr2b}
\Delta_{0}&=r_{0}^{2}-2GMy^{2}r_{0}+y^{2}a^{2},\\
\Sigma_{0}&=\left(r_{0}^{2}+y^{2}a^{2}\right)^{2}-y^{2}a^{2}\Delta_{0}.\\
\end{split}
\end{equation}

\noindent The circular orbit in the equatorial plane (see, e.g., \cite{bardeen}, \cite{BS})  with $d\phi/dt >0$ at radial coordinate $r_{0}$, is given by,

\begin{equation}
\label{schwarz2a}
\sigma(t)=\left(t, r_{0}, \pi/2, \frac{t}{y^{2} a + \sqrt{\frac{r_{0}^{3}}{GM}}}\right),
\end{equation}

\noindent When $a>0$, this orbit is co-rotational and when $a<0$ the orbit is retrograde. All circular orbits in the equatorial plane are stable for $ r_{0} > 9GMy^{2}$, independent of $a$, but the co-rotational circular orbit is stable for smaller values, even inside the ergosphere at $r_{0}= 2GMy^{2}$ when $a$ is sufficiently close to $GMy^{2}$ (cf. \cite{bardeen}).\\

\noindent The Killing field tangent to the path (where it is the four velocity) is given by,

\begin{equation}
\label{T32}
\vec{K}=\left(\frac{y^{2} a\sqrt{GM}+ r_{0}^{(3/2)}}{\mathcal{D}},0,0,\frac{\sqrt{GM}}{\mathcal{D}}\right),
\end{equation}

\noindent where,

\begin{equation}
\label{T33}
\mathcal{D}=\sqrt{2y^{2} a\sqrt{GM}r_{0}^{(3/2)}+r_{0}^{3}-3y^{2}GMr_{0}^{2}}.
\end{equation}

\noindent  Using Eqs.\eqref{Kerr1} and \eqref{T32}, $\alpha(y)=y\sqrt{-K^{\alpha}K_{\alpha}}$ may be found in Boyer-Lindquist coordinates.  Straightforward calculations show that $\alpha(0)=1$ and $\alpha'(0)=0$ and,

\begin{equation}
\label{T29}
\frac{1}{2}m\alpha''(0)=-\frac{GMm}{r}-\frac{GMmr^{2}\sin^{2}\theta}{2r_{0}^{3}}+\frac{3GMm}{2r_{0}}.
\end{equation}

\noindent This expression may be interpreted as Newtonian potential energy of a single particle in a container of particles in circular orbit with angular velocity,

\begin{equation}\label{angvel}
\dot{\phi}\equiv\frac{d\phi}{dt}=\sqrt\frac{GM}{r_{0}^{3}},
\end{equation}

\noindent where $r_{0}$ is the radius of the orbit. The first term on the right hand side of Eq.\eqref{T29} is the gravitational potential energy, and the third term is an additive constant that forces $\alpha''(0) =0$ when $r=r_{0}$ and $\theta = \pi/2$, i.e., at the origin of coordinates for the rotating container.  The second term is the centrifugal potential energy, more readily recognized when expressed in cylindrical coordinates.  To that end, let $\tilde{\rho} = r\sin\theta$,  with $\phi$ the azimuthal angle, and let $z$ measure linear distance along the axis of rotation.  The magnitude of angular momentum of a uniformly rotating particle of mass $m$ with cylindrical coordinates $(\tilde{\rho},\phi,z)$ is $\ell=m\tilde{\rho}^{2}\dot{\phi}$, and the centrifugal potential energy is then given by,

\begin{equation}\label{angvel2}
-\frac{\ell^{2}}{2m\tilde{\rho}^{2}}= -\frac{GMm\tilde{\rho}^{2}}{2r_{0}^{2}}=-\frac{GMmr^{2}\sin^{2}\theta}{2r_{0}^{3}}.
\end{equation}

\noindent The minus signs in Eq.\eqref{angvel2} take into account the direction of force, away from the central mass.\\ 

\noindent  In order to compute $\alpha''(0)$ in Fermi coordinates, we select the following tetrad vectors in the tangent space at $\sigma(0)$.

\begin{align}
e_{0}&=\left(\frac{y^{2} a\sqrt{GM}+r_{0}^{(3/2)}}{\mathcal{D}},0,0,\frac{\sqrt{GM}}{\mathcal{D}}\right),\notag\\
e_{1}&=\left (0,\frac{\sqrt{\Delta_{0}}}{r_{0}},0,0\right ),\notag\\
e_{2}&=\left (0,0,\frac{1}{r_{0}},0\right )\label{tt1},\\
e_{3}&=\left (\!\frac{y^{2}\sqrt{GM}\!\left(y^{2}a^{2}\!+\!r_{0}^{2}\!-\!2y^{2} a\sqrt{GMr_{0}}\right)}{\mathcal{D}\sqrt{\Delta_{0}}},0,0,\!\frac{y^{2}\sqrt{GM}(a -2\sqrt{r_{0}})+r_{0}^{(3/2)}}{\mathcal{D}\sqrt{\Delta_{0}}}\!\right)\notag
\end{align}

\noindent This tetrad may be extended via parallel transport to the entire circular orbit given by Eq.\eqref{schwarz2a}, but we need these tetrad vectors only at $\sigma(0)$.\\ 

\noindent Fermi coordinates relative to these coordinate axes (i.e., the above tetrad) may be calculated using Eq.(27) of \cite{KC1}. The result for the Boyer-Lindquist coordinates $r$ and $\theta$ to second order expressed in the Fermi space coordinates $x^{1},x^{2},x^{3}$ at ($x^{0} = \tau = t=0$) is,
\begin{eqnarray}
\theta&=&\frac{\pi}{2}+\frac{x^{2}}{r_{0}}-\frac{\sqrt{\Delta_{0}}\,x^{1}x^{2}}{r_{0}^{3}}+ \cdots,\label{tr3}\\
r&=&r_{0}+\frac{\sqrt{\Delta_{0}}\,x^{1}}{r_{0}}+\frac{\left(r_{0}y^{2}GM-y^{2}a^{2}\right)\left(x^{1}\right)^{2}}{2r_{0}^{3}}\nonumber\\
&\quad&+\frac{\Delta_{0}\left(x^{2}\right)^{2}}{2r_{0}^{3}}+\frac{\left(r_{0}-y^{2}GM\right)\left(x^{3}\right)^{2}}{2r_{0}^{2}}+\cdots\label{tr3'},
\end{eqnarray}

\noindent Now, calculating $\alpha(y)=y\sqrt{-K^{\alpha}K_{\alpha}}$ in Boyer-Lindquist coordinates, substituting for $r$ and $\theta$ using Eqs.\eqref{tr3} and \eqref{tr3'} gives,

\begin{eqnarray}
\begin{split}\label{alpha}
\alpha(y) =1&+\frac{GMy^{2}\left(r_{0}^{2}+3\,a^{2}y^{2}-4\,ay^{2}\sqrt{r_{0}\,GM}\right)\left(x^{2}\right)^{2}}{2\,r_{0}^{2}\,\mathcal{D}^{2}}\\
&-\frac{3\,GMy^{2}\Delta_{0}\left(x^{1}\right)^{2}}{2\,r_{0}^{2}\,\mathcal{D}^{2}}+O(3)
\end{split}
\end{eqnarray}

\noindent Computing the second derivative with respect to $y$ at $y=0$ yields,

\begin{equation}
\label{T30}
\begin{split}
\frac{1}{2}&m\alpha''(0)=\left(-\frac{GMm}{r_{0}}+\frac{GMm\,x^{1}}{r_{0}^{2}}-\frac{GMm\left(2\left(x^{1}\right)^{2}-\left(x^{2}\right)^{2}-\left(x^{3}\right)^{2}\right)}{2r_{0}^{3}}\right)\\
&-\left(\frac{GMm}{2r_{0}}+\frac{GMm\,x^{1}}{r_{0}^{2}}+\frac{GMm\left(\left(x^{1}\right)^{2}+\left(x^{3}\right)^{2}\right)}{2r_{0}^{3}}\right)+\frac{3GMm}{2r_{0}}+O(3).
\end{split}
\end{equation}

\noindent Eq.\eqref{T30} may be compared term-by-term with Eq.\eqref{T29}. The expression in the first pair of parentheses on the right hand side of Eq.\eqref{T30} is the Taylor expansion to second order of the gravitational potential, which is the first term on the right hand side of Eq.\eqref{T29}.  The second terms in both equations are related analogously.  At the point $\sigma(0)$ in the orbit, the Cartesian variable $x^{1}$ in Eq.\eqref{T30} measures (Newtonian) distance from the origin of coordinates in the radial direction away from the central mass, $x^{2}$ measures distance in the ``z direction'' parallel to the axis of rotation, and $x^{3}$ measures distance in the tangential direction, parallel to the motion of the container of gas in orbit. (However, these orientations do not hold at other parts of the orbit since the coordinate axes are nonrotating.)  Eq.\eqref{T30} may obviously be simplified to yield the potential energy function,

\begin{equation}
\label{T31}
U(x^{1},x^{2},x^{3})=\frac{1}{2}m\alpha''(0)=-\frac{3GmM\left(x^{1}\right)^{2}}{2r_{0}^{3}}+\frac{GmM\left(x^{2}\right)^{2}}{2r_{0}^{3}}+O(3).
\end{equation}

\noindent Combining Eq.\eqref{T31} with Eqs.\eqref{znewton} and \eqref{znewton2} immediately yields the Newtonian pressure (with $y =1/c$) for a gas in circular orbit around a central mass to which the relativistic counterpart given by Eq.\eqref{idealgas1} may be compared.\\

\section*{\normalsize{7. Anti-de Sitter space}}

In this section, we prove uniqueness of the infinite volume Gibbs state in anti-de Sitter space for a class of equilibrium interaction potentials and calculate the infinite volume relativistic pressure for an ideal gas of test particles.\\

\noindent Explicit transformation formulas to and from Fermi coordinates $\{x^{0}, x^{1}, x^{2}, x^{3}\}$ for a geodesic observer in (the covering space for) anti-de Sitter space-time are given in \cite{KC3}. The geodesic path of the Fermi observer in these coordinates is $\sigma(t)=(t,0,0,0)$.  Fermi coordinates for this space-time are global, i.e., one coordinate patch covers the entire space-time.\\  

\noindent Under the change of coordinates, $x^{0}=t$, $x^{1}=\rho \sin\theta \cos\phi$, $x^{2}=\rho \sin\theta \sin\phi$, $x^{3}=\rho \cos\theta$, the metric becomes diagonal. In these ``polar-Fermi coordinates,'' for any appropriate range of the angular coordinates, the line element becomes, 

\begin{equation}\label{polarfermi}
ds^{2}= -c^{2}\cosh^{2}\left( a\rho \right) dt^{2}+d\rho^{2}+\frac{\sinh^{2}(a\rho)}{a^{2}}(d\theta^{2} + \sin^{2} \theta \,d\phi^{2}),
\end{equation}

\noindent where $a =\sqrt{ |\lambda|/3}$ and $\lambda <0$ is the cosmological constant.  For this space-time, $\vec{K}=\partial/\partial t$ is a timelike Killing vector field tangent to $\sigma$ so that $e_{0} = \vec{K}$. From Eqs.\eqref{gamma} and \eqref{polarfermi},

\begin{equation}\label{newgamma2}
\gamma(x^{1}, x^{2}, x^{3})=\|\vec{K}\|\beta mc =\beta m c^2\cosh(a \rho),
\end{equation}

\noindent  where $\rho=\sqrt {(x^1)^2+(x^2)^2+(x^3)^2}$. The Fermi surface, $X$, given by Eq.\eqref{slice2}, is hyperbolic space.\\

\noindent Because a single chart, with domain $\mathbb{R}^{3}$, covers the entire Fermi surface for this space-time, each configuration of particles on the Fermi surface corresponds to a configuration on $\mathbb{R}^{3}$, and Gibbs states on the Fermi surface are in a natural one-to-one correspondence with Gibbs states of configurations on $\mathbb{R}^{3}$ with Poisson measure of Eq.\eqref{poisson1} based on Lebesgue measure $d\mathbf{x}$.  For an elaboration, see \cite{rockner, albeverio, kuna}. In this context, and noting Remark \ref{existence}, we may regard $V$ as an interaction on configurations on $\mathbb{R}^{3}$ and apply the following special case of Theorem 1 of \cite{KY2} to the present circumstances.\\ 

\begin{theorem}\label{old}
Suppose there exists an increasing sequence of bounded Borel sets $\{\Lambda_{k}\}$ whose union is $\mathbb{R}^{3}$ such that:
\begin{enumerate}
\item [(1)]$\phi_{n}(x) =0$, for any $n\geq2$ and $x=(x_{1}, ..., x_{n})$ such that $x\cap\Lambda_{k}\neq\varnothing\neq x\cap(\mathbb{R}^{3}/\Lambda_{k+1})$ for some k;
\item[(2)] $\sum_{k}\left[\sup_{s}Z_{A_{k}}(s, \beta, z)\right]^{-1}$ diverges, where $A_{k}=\Lambda_{k+1}/\Lambda_{k}$.
\end{enumerate}
Then there is at most one Gibbs state for $V,\beta, z$.
\end{theorem}

\noindent Assume that $V$ is a superstable interaction and has the form given by Eqs.\eqref{1body} and \eqref{V2}.  Assume further that V has finite range and satisfies the condition,

\begin{equation}\label{stable}
\tilde{V}_{\Lambda}(x_{1}, ..., x_{n}|s)\equiv \sum_{n=2}^{\infty}\, \sum _{\substack{y \subset x\cup (s\cap X/\Lambda)\\
y\cap x\neq\varnothing\\|y|=n}}
\phi_{n}(y)>-Bn,
\end{equation}

\noindent for some $B>0$, all boundary configurations $s$, and all configurations $(x_{1}, ..., x_{n})$.
 Observe that $\tilde{V}$ is the same as $V$ except that $\tilde{V}$ does not include the vacuum energy $\phi_{0}$ or the one-body potential function, $\phi_{1}$.  The potential $V$ satisfies Eq.\eqref{stable} if it is positive or satisfies a hard-core condition (see \cite{KY2}).\\ 
 
 \begin{theorem}\label{uniqeness}
 Let $V$ be a potential energy function with finite range $R$ satisfying Eq.\eqref{stable} on the Fermi surface $X$ in anti-de Sitter space.  Then there is at most one Gibbs state on $(\Gamma_{X}, \mathcal{B}(\Gamma_{X}))$ for $V, \beta, z$ for all choices of $\beta$ and $z$.
 \end{theorem}
 
 \begin{proof}
 Let $\Lambda_{k}\subset X$ be the sphere of radius $kR$ centered at the origin of coordinates in the chart for X, where $R$ is as in Def. \ref{finiterange}.  Then condition $(1)$ of Theorem \ref{old} is satisfied by $\{\Lambda_{k}\}$.  To see that condition $(2)$ of Theorem \ref{old} is also satisfied, let $A_{k}$ be defined as in Theorem \ref{old} and observe by Eq.\eqref{stable} that,

\begin{equation}\label{bound}
\sup_{s} Z_{A_k} (s, \beta, z) \le \int_{\Gamma_{A_{k}}} \exp\left[\beta Bn  -\beta \sum_{i=1}^{|x|} \phi _{1}(x_{i})\right]\,\nu_{A_{k}}(dx).
\end{equation}

 \noindent It follows from Eqs.\eqref{asymp} and \eqref{newgamma2} that when  $\rho=\sqrt {(x^1)^2+(x^2)^2+(x^3)^2}$ is sufficiently large,
 
 \begin{equation}\label{asymp2}
e^{- \beta \phi _1 (x^1, x^2, x^3)}\leq e^{-\beta m c^2\cosh(a \rho)}.
\end{equation}

 \noindent Then for large $k$,
 
 \begin{equation}\label{bound2}
\sup_s Z_{A_k} (s, \beta, z) \le \sum_{n=0}^\infty \exp [\beta B n - \beta n \cosh(kR) ] \frac{z^n}{n!}  |A_{k}|^{n}, 
\end{equation}

\noindent where $|A_{k}|=\int _{A_{k}} d x^1 dx^2 dx^3$, i.e., the Lebesgue measure of the image of $A_{k}$ in the Fermi coordinate chart.  Hence, $|A_{k}|< C(kR)^{2}$ for a constant $C$. Thus,

\begin{equation}\label{bound3}
\sup_s Z_{A_k} (s, \beta, z) \le \exp [zC(Rk)^{2} e^{\beta B  - \beta  \cosh(kR)} ], 
\end{equation}

 \noindent and condition $(2)$ of Theorem \ref{old} is clearly satisfied.
 \end{proof}

\begin{corollary}\label{adspressure}
The pressure of an ideal gas of test particles in anti-de Sitter space is given by, 
\begin{equation}\label{adspressure}
P=-\frac{\lambda c^{4}}{8\pi G},
\end{equation}
\noindent where $\lambda<0$ is the cosmological constant for anti-de Sitter space.
 \end{corollary}

\begin{proof}
With the same notation as in the proof of Theorem \ref{uniqeness}, using Eq.\eqref{idealgas1} and changing to spherical coordinates, we find,

\begin{equation}\label{thm3b}
P_{\Lambda_{k}}(\beta, z)= -\rho_{\text{vac}}+\frac{4\pi z}{\beta mc^{2}|\Lambda_{k}|}\int_{0}^{k}\frac{K_{2}(\beta m c^2\cosh(a \rho))}{\cosh(a \rho)}\rho^{2}d\rho.
\end{equation}

\noindent The result now follows from Eqs.\eqref{pressure}, \eqref{vacuumenergy}, and \eqref{asymp}.

\end{proof}

\begin{remark}\label{nonewton} The infinite volume pressure of an ideal gas of test particles in anti-de Sitter space is independent of temperature and chemical activity, and depends only on the magnitude of the cosmological constant. There is no corresponding Newtonian analog, as developed in Sect. 5, because Eq.\eqref{alpha2} does not hold for the Killing vector $\partial/\partial t$.
\end{remark}

\section*{\normalsize{8. Concluding Remarks}}

For space-times with timelike Killing fields, we have developed a grand canonical formalism for statistical mechanical systems of test particles.    The thermodynamic behavior of an ideal gas was shown in Sect. 4 to be determined by the metric of the space-time through the norm of the timelike Killing field, and the volume of the particle system as determined by projection $^{3}g$ of the metric on the Fermi space slice of the timelike observer.\\ 

\noindent We proved uniqueness of the Gibbs state for a class of interaction potentials in anti-de Sitter space, and found the infinite volume pressure for an ideal gas of test particles. For the case of an ideal gas in space-times where Eq.\eqref{alpha2} is satisfied, we introduced a notion of Newtonian limit of Gibbs states and of grand canonical pressures.  The physical legitimacy of the grand canonical formalism was demonstrated in Sect. 6, where it was shown that the Newtonian limit of a relativistic particle system following a circular geodesic orbit in Kerr space-time results in exactly what is predicted by classical mechanics.\\  

\noindent Directions for further research include possible  generalizations of the statistical mechanical formalism which assume only the existence of approximate timelike Killing vector fields or conformal timelike Killing fields, relaxing our assumption of the existence of a timelike Killing vector field.  In addition, alternative notions of simultaneity associated with  coordinate systems other than Fermi-Walker-Killing coordinates might be considered, such as standard curvature-normalized coordinates on Robertson-Walker cosmologies for observers co-moving with the Hubble flow, or optical coordinates (in which events are ``simultaneous'' if they are visible to the observer $\sigma$ at the same proper time coordinate of $\sigma$).  In this way, the effect of the curvature of space on statistical mechanical behavior of test particles might be analyzed for realistic circumstances.\\

\noindent \textbf{Acknowledgment.} The authors wish to thank Peter Collas for helpful discussions.


\begin{thebibliography}{99}

\bibitem{rockner} R\"{o}ckner, M.: Stochastic analysis on configuration spaces: basic ideas and recent results, {\it New directions in Dirichlet forms}, 157-231, AMS/IP Stud. Adv. Math., \textbf{8}, Amer. Math. Soc., Providence, RI, (1998) arXiv:math/9803162v1 [math.PR] 

\bibitem{albeverio} Albeverio, P., Kondratiev, Y., R\"{o}ckner,
M.: Analysis and Geometry on Configuration Spaces: 
The Gibbsian Case {\it J. Functional Analysis}  \textbf{157}, 242-291 (1998).

\bibitem{kuna} Kuna, T.: Studies in Configuration Space Analysis and Applications,  \textit{Bonner Math. Schrift.}, \textbf{324}, Univ. Bonn, Bonn, 1999, PhD dissertation, Rheinische Friedrich-Wilhelms-Universit\"at Bonn, Bonn (1999), 187 pp. MR1932768 (2003k:82058).

\bibitem{kondratiev} Kondratiev, Y., Pasurek, T., R\"{o}ckner,
M.: Gibbs Measures of Continuous Systems: an analytic approach www.math.uni-bielefeld.de/sfb701/preprints/sfb09019.pdf (2009).

\bibitem{KC2} Klein, D., Collas, P.: Timelike Killing fields and relativistic statistical mechanics, {\it Class. Quantum Grav.}  \textbf{26}, 045018 (16 pp) arXiv:0810.1776v2 [gr-qc] (2009).

\bibitem{rovelli} Montesinos, M., Rovelli, C.: Statistical mechanics of generally covariant quantum theories: a Boltzmann-like approach \textit{Class. Quantum Grav.} \textbf{18}, 555-569 (2001).

\bibitem{MTW} Misner, C. W., Thorne, K. S., and Wheeler, J. A. (1973). \textit{Gravitation}, W. H. Freeman, San Francisco.

\bibitem{walker} Walker, A. G.: Note on relativistic mechanics  \textit{Proc. Edin. Math. Soc.} \textbf{4}, 170-174 (1935).

\bibitem{synge}  Synge, J. L.: \textit{Relativity: The General Theory} North Holland, Amsterdam (1960).

\bibitem{CK2} Collas, P., Klein, D.: A Simple Criterion for Nonrotating Reference Frames, \textit{Gen. Rel. Grav.} \textbf{36}, 1493-1499  (2004).

\bibitem{KC3}  Klein, D., Collas, P.: Exact Fermi coordinates for a class of space-times, \textit{J. Math. Phys.}  \textbf{51}, 022501, (10 pp) DOI:10.1063/1.3298684, arXiv:0912.2779v1 [math-ph] (2010).

\bibitem{oniell} O'Neill, B.: \textit{Semi-Riemannian geometry with applications to relativity} (1983).  Academic Press, New York, p. 200.

\bibitem{KC1}  Klein, D., Collas, P., General Transformation Formulas for Fermi-Walker Coordinates \textit{Class. Quant. Grav.}  \textbf{25}, 145019 (17pp) DOI:10.1088/0264-9381/25/14/145019, [gr-qc] arxiv.org/abs/0712.3838v4 (2008).

\bibitem{bolos} Bol\'os, V.: Intrinsic definitions of ``relative velocity'' in general relativity \textit{Comm. Math. Phys.}  \textbf{273},  217-236 (2007).

\bibitem{Klein11} Klein, D., Randles, E.,  Fermi coordinates, simultaneity, and expanding space in Robertson-Walker cosmologies \textit{Ann. Henri Poincar\'e} \textbf{12} 303--28 (2011) DOI: 10.1007/s00023-011-0080-9.

\bibitem{bolos2} Bol\'os, V., Klein D., Relative velocities for radial motion in expanding Robertson-Walker spacetimes  Archive: arXiv:1106.3859v1 [gr-qc]

\bibitem{CK} Collas, P., Klein, D., A Statistical mechanical problem in Schwarzschild space-time \textit{Gen. Rel. Grav.} \textbf{39}, 737-755 DOI:10.1007/s10714-007-0416-4 (2007).

\bibitem{tolman} Tolman, R. C.: Relativity, thermodynamics, and cosmology. Clarendon, Oxford (1934) pp. 318-9.

\bibitem{special} Horwitz, L. P., Schieve, W. C., Piron, C.: Gibbs ensembles in relativistic classical and quantum mechanics \textit{Ann. Phys., NY} \textbf{137}, 306-340 (1981).

\bibitem{E53} Erd\'{e}lyi, A. (1953)  \textit{Higher Transcendental Functions Vol. II}.  McGraw-Hill, New York, p. 23.

\bibitem{prob} Parthasarathy, K. R., \textit{Probability measures on metric spaces}, Probability and mathematical statistics, Academic Press, New York and London (1967).

\bibitem{boehmer} B\"{o}hmer, C.G., Harko, T.: Does the cosmological constant imply the existence of a minimum mass? \textit{Phys. Lett. B} \textbf{630}, 73-77 (2005) [arXiv:gr-qc/0509110]. 

\bibitem{bardeen}  Bardeen, J. M., Press, W. H., Teukolsky, S. A., Rotating Black Holes: Locally Nonrotating Frames, Energy Extraction, and Scalar Synchrotron Radiation  \textit{The Astrophysical Journal}  \textbf{178},  347-369 (1972).

\bibitem{BS}  Bonnor, W., Steadman, B., The gravitomagnetic clock effect \textit{Class. Quant. Grav.}  \textbf{16},  1853--1861(1999).

\bibitem{KY2} Klein, D., Yang, W.S.:  Absence of phase transitions for continuum models of dimension $d>1$  {\it J. Math. Phys} \textbf{29}, 1130-1133 (1988).

%\bibitem{CM} Chicone, C., Mashhoon, B.: Explicit Fermi coordinates and tidal dynamics in de Sitter and G\"odel space-times \textit{Phys. Rev.} D \textbf{74},  064019 (2006).








\end{thebibliography}
\end{document}